\documentclass[journal,comsoc]{IEEEtran}
\usepackage[T1]{fontenc}
\usepackage{graphicx}

\graphicspath{ {images/} }
\usepackage{algorithmic}
\usepackage{algorithm}
\usepackage{amsmath}
\usepackage{amssymb}
\usepackage{textcomp}
\usepackage{amsthm}
\usepackage{wrapfig}
\usepackage{lipsum}
\usepackage{xcolor}

\usepackage{soul}

\graphicspath{ {images/} }
\usepackage{array}
\ifCLASSOPTIONcompsoc
  \usepackage[caption=false,font=normalsize,labelfont=sf,textfont=sf]{subfig}  
\else
  \usepackage[caption=false,font=footnotesize]{subfig}
\fi
\usepackage{stfloats}
\usepackage{url}
\usepackage{balance}
\newtheorem{theorem}{Theorem}
\newtheorem{definition}{Definition}
\newtheorem{remark}{Remark}

\usepackage[cmintegrals]{newtxmath}

\begin{document}

\title{Multivariate Extreme Value Theory Based Channel Modeling for Ultra-Reliable Communications}

\author{Niloofar~Mehrnia,~\IEEEmembership{Student Member,~IEEE,}
        Sinem~Coleri,~\IEEEmembership{Fellow,~IEEE}
\thanks{Copyright (c) $2015$ IEEE. Personal use of this material is permitted. However, permission to use this material for any other purposes must be obtained from the IEEE by sending a request to pubs-permissions@ieee.org.}
\thanks{Niloofar Mehrnia and Sinem Coleri are with the Department of Electrical and Electronics Engineering,
Koc University, Istanbul, Turkey (e-mail: nmehrnia17@ku.edu.tr; scoleri@ku.edu.tr).}
\thanks{Niloofar Mehrnia is also with Koc University Ford Otosan Automotive Technologies Laboratory (KUFOTAL), Sariyer, Istanbul, Turkey, 34450.}
\thanks{Sinem Coleri acknowledges the support of Ford Otosan, and the Scientific and Technological Research Council of Turkey 2247-A National Leaders Research Grant $\#121C314$.}
}

\maketitle

\begin{abstract}
Attaining ultra-reliable communication (URC) in fifth-generation (5G) and beyond networks requires deriving statistics of channel in ultra-reliable region by modeling the extreme events. Extreme value theory (EVT) has been previously adopted in channel modeling to characterize the lower tail of received powers in URC systems. In this paper, we propose a multivariate EVT (MEVT)-based channel modeling methodology for tail of the joint distribution of multi-channel by characterizing the multivariate extremes of multiple-input multiple-output (MIMO) system. The proposed approach derives lower tail statistics of received power of each channel by using the generalized Pareto distribution (GPD). Then, tail of the joint distribution is modeled as a function of estimated GPD parameters based on two approaches: logistic distribution, which utilizes logistic distribution to determine dependency factors among the Fréchet transformed tail sequence and obtain a bi-variate extreme value model, and Poisson point process, which estimates probability measure function of the Pickands angular component to model bi-variate extreme values. Finally, validity of the proposed models is assessed by incorporating the mean constraint on probability measure function of Pichanks coordinates. Based on the data collected within the engine compartment of Fiat Linea, we demonstrate the superiority of proposed methodology compared to the conventional extrapolation-based methods in providing the best fit to the multivariate extremes.
\vspace{0.4cm}

\end{abstract}


\section{Introduction}
\label{sec:intro}
Ultra-reliability is one of the most important features of the fifth generation ($5$G) and beyond networks with the goal of providing a reliable and resilient communication infrastructure that can support a wide range of applications and use cases, including industrial automation, self-driving cars, and remote medical procedures \cite{interface_01}-\nocite{5G_01}\nocite{interface_02}\cite{urllc_02}. In the current terminology of $5$G standards, ultra-reliability is commonly associated with low latency and creates an ultra-reliable low latency communication (URLLC). However, this tight coupling should be relaxed in certain scenarios such as health monitoring or disaster recovery applications where ultra-reliability is essential, but the allowed latency can be larger than the typical $1$~milliseconds (ms) \cite{urllc_05}.
Fulfillment of the reliability requirements of $10^{-9}$-$10^{-5}$ packet error rate (PER) in ultra-reliable communication (URC) has been widely used in previous studies \cite{urllc_02}-\nocite{urllc_05}\nocite{reliability_01}\nocite{urc_03}\nocite{urc_04}\nocite{urc_05}\cite{urc_06}. URC necessitates fundamental breakthrough in the derivation of the channel lower tail distribution by using novel statistical tools, such as extreme value theory (EVT) \cite{MehrniaTWC}, \cite{MehrniaTVT}, power law expression based on the extrapolation of the channel data \cite{urllc_02}, \cite{urllc_05}, and data-driven learning framework \cite{reliability_01}.
Different forms of diversity techniques can be applied in URC to reduce the required signal-to-noise-ratio (SNR) for achieving a certain reliability. Among the diversity techniques, spatial diversity is the most viable solution to achieve high reliability over fading channels by utilizing multiple transmission antennas \cite{urllc_diversity_01}. Consequently, towards achieving ultra-reliability in spatial diversity, the key is to obtain a proper multi-channel modeling approach that characterizes the statistics of multivariate extreme events for the systems with multiple input multiple outputs (MIMO) connections \cite{evt_04}.

Ultra-reliability can be achieved through different forms of diversity techniques, including time diversity \cite{timediv_01}, frequency diversity \cite{urllc_diversity_02}, interface diversity \cite{interferencediv}, or spatial diversity \cite{urllc_02}, \cite{urllc_08}, \cite{urllc_diversity_03}-\nocite{urllc_diversity_05}\cite{urllc_diversity_04}.
In \cite{timediv_01}, a combination of the stochastic geometry and the queueing analysis is proposed to derive the URLLC achievable ratio through the concept of effective bandwidth, which transforms the constraints of delay and reliability into the constraint of the achievable transmission rate. The authors in \cite{urllc_diversity_02} explore the diversity aspects of random access schemes where users transmit over multiple orthogonal Rayleigh fading subchannels and are treated as diversity branches by a maximum ratio-combining (MRC) receiver.  
The authors in \cite{interferencediv} consider a trade-off between transmission latency and reliability by allocating coded fragments of the encoded payload message to different interfaces according to their bit rate, latency, and reliability properties. \cite{urllc_diversity_03} proposes a statistical model by incorporating the time of arrival (TOA), angle of arrival (AOA), and angle of departure (AOD) of each multi-path component to estimate the joint magnitude and phase PDFs of measured data for realistic model parameters. 
A survey presented in \cite{urllc_diversity_05} investigates the main properties of massive MIMO channels based on different configurations of an antenna array, which directly affect the channel models and system performance. On the other hand, \cite{urllc_diversity_04} provides a survey of the most critical concepts in propagation channel modeling for MIMO systems by classifying the channels into physical models focusing on double-directional propagation and analytical models considering the channel impulse response and antenna properties.
Nevertheless, an essential building block of an ultra-reliable wireless system is a model of the wireless channel at the physical layer that captures the statistics of rare events, and extensive fading dips \cite{urllc_02}, \cite{urllc_08}.

Current studies on the statistical modeling of the wireless channel operating in the ultra-reliable communication region are categorized into four folds: $a$) Extrapolation of a wide range of commonly used average statistics-based wireless channel models such as Rayleigh and Rician towards the ultra-reliability region \cite{urllc_02}, \cite{urllc_05}; $b$) Recommendation of the usage of new channel parameters by incorporating extreme reliability requirements into the communication \cite{urllc_07}, \cite{urllc_08}; $c$) Usage of a non-parametric statistical learning algorithm to estimate the probability density function (PDF) of the channel distribution \cite{confidenceinterval_03}; and $d$) Application of the EVT to characterize the statistics of the channel tail distribution by deriving the statistics of extreme events happening rarely \cite{MehrniaTWC}, \cite{MehrniaTVT}. In \cite{urllc_02}, a simple power-law expression is proposed to extrapolate the cumulative distribution function (CDF) of the commonly used average statistics channels to the ultra-reliable regime of operation. In addition, the authors in \cite{urllc_02} apply the power-law results to analyze the performance of receiver diversity schemes and obtain a new simplified expression for MRC applicable for an ultra-reliable communication. In \cite{urllc_05}, two performance measures, average reliability (AR) for dynamic environments and probably correct reliability (PCR) for the static environments, have been proposed based on the extrapolation-based channels in \cite{urllc_02} with the goal of selecting a transmission rate guaranteeing ultra-reliability. However, these extrapolated distributions cannot accurately estimate the lower tail distribution, and therefore, result in several orders of magnitude differences in the estimated PER \cite{MehrniaTWC}, \cite{MehrniaTVT}. Also, the channels in \cite{urllc_02} are assumed independent and non-identically distributed (i.n.i.d.), which is not a realistic assumption.
In \cite{urllc_07}, \cite{urllc_08}, the authors have modified the definitions of coherence time and distance to the time and distance over which a channel is predictable with a given reliability, respectively. Accordingly, the authors in \cite{urllc_08} use a standard spatially independent quasi-static Rayleigh fading model of wireless channels to examine channel dynamics in the URC context. Nevertheless, none of these frameworks derives ultra-reliability statistics by proposing a channel modeling methodology. In \cite{confidenceinterval_03}, the authors propose a non-parametric statistical learning algorithm that utilizes kernel density estimation (KDE) to estimate the PDF of channel distribution and accordingly, selects a proper transmission channel rate that meets the requirements of ultra-reliability. However, the proposed data-driven non-parametric approach requires a massive number of training samples, about $10 \times \epsilon^{-1}$, to meet the targeted error probability $\epsilon$. 
Only recently in \cite{MehrniaTWC}, \cite{MehrniaTVT}, we have proposed a novel channel modeling methodology based on EVT to derive the statistics of the lower tail distribution while efficiently dealing with a massive amount of data. EVT is a unique and robust framework to develop techniques for modeling the statistics of rare events based on the implementation of mathematical limits as finite-level approximation \cite{urllc_02}, \cite{evt_04}, \cite{evt_03}.

EVT has been used at the data link and network layers to model the tail statistics of queue length and delay \cite{upperlayer_03}-\nocite{upperlayer_05}\cite{upperlayer_11} and derive closed-form asymptotic expressions for the throughput, bit error rate (BER), and PER over different fading channels \cite{upperlayer_04}-\nocite{upperlayer_02}\nocite{upperlayer_01}\cite{upperlayer_12}. Moreover, block maxima EVT models have been applied to model the maximum end-to-end (E2E) latency by using the generalized extreme value (GEV) distribution for virtual reality (VR) applications in Terahertz \cite{urllc_10}. 
Furthermore, EVT has been used to derive the distribution of the maximum end-to-end SNR in an opportunistic relay selection-based cooperative relaying network consisting of a large number of independent and non-identical relay links \cite{upperlayer_10}.
Additionally, EVT and federated learning (FL) approaches have been combined in \cite{upperlayer_13} to propose a Lyapunov-based distributed transmit power and
resource allocation procedure for vehicular users (VUEs), while the statistics of the queue lengths exceeding a high threshold are characterized by using the GPD.
However, these upper layer studies assume the average statistic-based channel models such as Rayleigh, Rician, or Nakagami fading; therefore, their usage may not be suitable in a system operating at URC \cite{urllc_02}, \cite{MehrniaTWC}, \cite{MehrniaTVT}. To fill this gap, in \cite{MehrniaTWC}, we have used EVT techniques at the physical layer to determine the optimum threshold below which the received power samples are considered extreme events and included in the tail distribution, then model the channel tail distribution by using the generalized Pareto distribution (GPD), and finally, assess the validity of the proposed EVT-based model by using the probability plots. Additionally, in \cite{MehrniaTVT}, we have extended \cite{MehrniaTWC} to model the tail of the non-stationary channel distribution by first determining an external factor causing non-stationarity and, accordingly, dividing the channel data into multiple stationary sequences to apply the EVT techniques initially presented in \cite{MehrniaTWC} and model the parameter of the fitting distribution as a change-point function of time. 
Additionally, recently, in \cite{MehrniaTVTRate}, \cite{MehrniaEucnc}, we have proposed a novel EVT-based framework dealing with a relatively low number of data samples to estimate the optimal transmission rate and validate it by assessing the outage probability so that reliability constraints are met with given confidence for ultra-reliable communications. Nevertheless, we have not considered spatial diversity to address the reliability constraints in more than one dimension and assess the dependency between extreme events. Therefore, there is a lack of studies focusing on channel modeling strategies in the ultra-reliable regime of operation by discovering the inter-relation of extreme events. Multivariate extreme value theory (MEVT) is a robust statistical discipline that develops techniques to model the relation of rare events based on the multidimensional limiting relations \cite{urllc_02}, \cite{evt_04}, \cite{evt_03}.

The goal of this paper is to propose a novel channel modeling methodology based on MEVT for a system using spatial diversity in MIMO-URC to derive the lower tail statistics of the received signal power in multiple dimensions while efficiently dealing with a massive amount of corresponding data.
We focus on the two-dimensional or bi-variate case to highlight the main concepts and issues of MEVT without increasing the complexity of the notation required for an entire multivariate perspective.
The modeling approach adapts MEVT to 
(i) fit GPD to the tail distribution of the received power samples exceeding a given threshold in each sequence and derive the scale and shape parameters of GPD while the thresholds are determined optimally; 
(ii) apply Fr\'echet transformation to each data sequence so that the marginal distributions of the resulting variables are Fr\'echet;
(iii) determine the dependency factor between the Fr\'echet transformed sequences;
(iv) apply the logistic distribution approach and Poisson point process approach to fit bi-variate GPD (BGPD) to the tail of the joint distribution of the Fr\'echet transformed sequences, considering the dependency factor; 
and (v) assess the validity of the fitted BGPD model by checking the mean constraint on the probability measure function. MEVT extends the ideas of univariate EVT to the analysis of dependent extreme events for multiple random variables by introducing several features of MEVT that are not available in EVT analysis, including modeling dependence structures and modeling the joint distribution of multiple extreme events. The original contributions of the paper are listed as follows:

\begin{itemize}
    \item We propose a comprehensive channel modeling methodology for a system operating at URC based on MEVT for the first time in the literature. The methodology consists of techniques for deriving the lower tail statistics of each channel data sequence by using Uni-variate GPD (UGPD), fitting BGPD to the tail of the joint probability distribution by using the logistic distribution-based and Poisson point process-based approaches, and assessing the validity of these two proposed models by incorporating the probability measure function of the Pichanks coordinates.
    \item We propose novel techniques based on the logistic distribution to fit BGPD to the tail of a joint probability distribution of channel data, for the first time in the literature. We first determine the dependency factor among the Fr\'echet transformed sequences, and then derive a closed-form expression for the BGPD model.
    These techniques are original and contribute to the advancement of statistical models for analyzing joint probability distributions of extreme channel samples.
    \item We propose novel techniques based on Poisson point process approach to represent the Fr\'echet transformed channel tail data using their pseudo-polar Pickands pairs, radial and angular components, for the first time in the literature. We use these representations to determine the probability measure function of the Pickands angular component and the corresponding BGPD model.
    \item We introduce a novel approach for assessing the validity of the derived BGPD models for the joint distribution of the channel tail based on the verification of the mean constraint on the corresponding probability measure function of the Pickands coordinates obtained from BGPD approaches, for the first time in the literature.
    \item Based on the data collected within the engine compartment of Fiat Linea using one transmitter and two receivers under various engine vibrations and driving scenarios, we demonstrate the superiority of the proposed methodology for deriving the tail statistics of multivariate extremes compared to the conventional extrapolation-based models of the average statistics channels to the ultra-reliable region, in terms of the modeling accuracy.
\end{itemize}

The rest of the paper is organized as follows. Section~\ref{sec:system_model} describes the system model and assumptions considered throughout the paper. Section~\ref{sec:background} describes the basics of uni-variate and multivariate EVT together with the theorems used in the development of the proposed multivariate channel modeling approach. Section~\ref{sec:methodology} presents the proposed channel modeling framework based on MEVT for characterizing the joint multi-channel tail distribution in the ultra-reliable region. Section~\ref{sec:numerical_results} provides the channel measurement setup and the performance evaluation of the proposed algorithm in determining the optimum threshold, fitting BGPD to the joint probability distribution of the extremes, and comparing the proposed methodology to the conventional methods in terms of the estimation accuracy. Finally, concluding remarks and future works are given in Section~\ref{sec:conclusions}.

\section{System Model}
\label{sec:system_model}
We consider MIMO for a single transmitter (Tx)-receiver (Rx) pair communicating over a stationary channel, i.e., the parameters of the channel distribution are fixed over a vast period. If the channel is non-stationary according to the Augmented Dickey-Fuller (ADF) test results \cite{MehrniaTVT}, the external factors causing time variation are determined such that the sequence is divided into $M$ groups, each of which can be considered stationary, as explained in detail in \cite{MehrniaTVT}. Then, the tail distribution of each stationary sequence is modeled by using the GPD. The GPD is used in EVT to estimate the tail of distribution by modeling the probabilistic distribution of the values exceeding a given threshold \cite{evt_04}, \cite{evt_03}, \cite{evt_01}\nocite{evt_02}\nocite{evt_05}\nocite{evt_06}\nocite{evt_07}-\cite{evt_08}. Please note that the same procedure can be applied to multiple transmitter and receiver pairs.
The transmit power is assumed to be fixed and known in advance. Therefore, estimating the received signal power is equivalent to estimating the squared amplitude of the channel state information \cite{urllc_05}, \cite{MehrniaTWC}.  

Ultra-reliability is defined as communication with target packet error probability in the range of $10^{-9}$-$10^{-5}$. We assume that the outage is the only source of packet error and is defined as the received power being less than a predefined threshold \cite{urllc_02}, \cite{urllc_11}, \cite{urllc_12}. Since the main focus is on modeling the channel behavior in the ultra-reliable region by employing the distributions, modeling and validation techniques adopted from extreme value theory in an offline manner, the delay required for collecting a large number of empirical samples is not considered in the first step of the study. The design of a real-time URLLC system based on these statistics requires the consideration of a limited amount of data requiring the inclusion of confidence intervals in the parameter estimation or usage of transfer learning techniques, but it is out of the scope of this paper and subject to future work.

In order to model the multivariate extremes of received powers, the sequences of measured received power samples at Tx-Rx pairs are converted into a sequence of independent and identical distributed (i.i.d.) samples by removing their dependency via declustering approach \cite{MehrniaTWC}. Upon applying EVT to the resulting sequence of i.i.d. samples in each pair, the optimum threshold is determined for each sample sequence. Determining an optimum threshold is of paramount importance as it determines the number of samples that are included in the channel tail and considered as an extreme value. Then, the parameters of the Pareto distribution associated with the optimum threshold are estimated by using the maximum likelihood estimator (MLE).   
A multivariate version of GPD, i.e., a family with which to approximate a joint distribution on regions where we observe joint extreme values, is obtained by using the MEVT.

\section{Background}
\label{sec:background}
EVT is a robust framework that models the probabilistic distribution of extreme events occurring rarely. EVT has been applied in different fields, including hydrology to quantify risks of extreme floods, rainfalls, and waves \cite{evt_03}, and finance to estimate losses due to extreme events \cite{evt_01}. However, it has been recently used in the telecommunication field to analyze the behavior of extreme values either in network traffic, worst-case delay, queue lengths, throughput, and BER/PER of URLLC \cite{MehrniaTVTRate}, \cite{upperlayer_03}, \cite{urllc_01}, or in channel modeling to statistically derive the lower tail of the received power for URC systems \cite{MehrniaTWC}-\cite{MehrniaTVT}. 
In the following, Section~\ref{sec:evt} presents the uni-variate technique for modeling the extremes of a single process. Section~\ref{sec:mevt} gives the general modeling technique of multivariate extremes to describe the extreme value inter-relationships of two or more processes. Sections~\ref{sec:backgroundlogistic} and \ref{sec:backgroundpointprocess} present bi-variate extreme value techniques based on the logistic family distribution and Poisson point process approach, respectively. Finally, Section~\ref{sec:background_assessment} provides techniques to assess the validity of the proposed bi-variate extreme value modeling techniques.

\subsection{Uni-variate Extreme Value Theory}
\label{sec:evt}
Uni-variate EVT (UEVT), in general, focuses on the representation and modeling techniques for the extremes of a single process. UEVT is used for modeling extreme events in two main ways. The first concerns models for block maxima by using the (GEV) distribution given by
\begin{equation}
\label{eqn:evt_dist}
    F(z) = \exp{\big\{-\big[1+\xi\big(\frac{z-\mu}{\sigma}\big)\big]^{-1/\xi}\big\}}
\end{equation}
where $\mu$, $\xi$, and $\sigma$ are the location, shape, and scale parameters of the GEV distribution, respectively \cite[Theorem~3.3]{evt_04}.
The second uses EVT to characterize the tail of a distribution, i.e., the extremes of a sequence, by modeling the probabilistic distribution of values exceeding a given threshold through the generalized Pareto distribution (GPD).

Assume that $\{x_1,...,x_N\}$ is an i.i.d. stationary sequence, where $x_{i}$ denotes the $i^{th}$ received power for $i \in \{1,...,N\}$, and $N$ is the total number of received power samples. Then, according to the EVT, the tail distribution of the power sequence, i.e., the probabilistic distribution of the power values exceeding a given threshold $u$, can be expressed as
\begin{equation}
\label{eqn:gpd_dist}
    G(l) = 1-\big[1+\frac{\xi l}{\tilde{\sigma}}\big]^{-1/\xi},
\end{equation}
where $l$ is a non-negative value denoting the exceedance below threshold $u$, i.e., ($l=u-X$), $X$ denotes any $x_i$ below threshold $u$; $G(l)$ expresses the GPD; and $\xi$ and $\tilde{\sigma}=\sigma+\xi(u-\mu)$ are shape and scale parameters of the GPD, respectively. It should be noted that $\mu$ and $\sigma$ are the location and scale parameters of the GEV distribution fitted to the CDF of $m_N = min \{ x_1,...,x_N\}$, respectively \cite[Theorem~1]{MehrniaTWC}, \cite{evt_11}.

\subsection{Bi-variate Extreme Value Theory}
\label{sec:mevt}
When studying the extremes of two or more processes, the individual process can be modeled by using uni-variate techniques. However, the possible dependency between the extreme events requires the investigation of their joint behavior. 
Bi-variate EVT (BEVT) allows us to estimate the probability of exceeding thresholds simultaneously and analyze the inter-dependence of two variables in the extreme value region \cite{mevt_11}. 

In the following, we provide the Theorems and Corollaries required to develop our channel modeling methodology. Theorem~\ref{theorem:bevt} incorporates BEVT applications to investigate the general form of BGPD models that are valid to estimate the bi-variate tail distribution.
Then, based on Theorem~\ref{theorem:bevt}, we obtain the logistic family distribution-based BGPD and Poisson point process-based BGPD models in Theorems~\ref{theorem:logistic} and \ref{theorem:pointprocess}, respectively. Also, we define the Fr\'echet transformation in Definition~\ref{def:frechet} as the input of BEVT in Theorem~\ref{theorem:bevt} is required to be Fr\'echet distributed. Additionally, the definition of Pseudo-polar Pickands coordinate transformation is provided in Definition~\ref{def:pickand} to be used in the Poisson point process approach in Theorem~\ref{theorem:pointprocess}. 
Moreover, the constraints on the Pickands coordinates are defined in Definition~\ref{def:pickandsconstraints}, which will be used to estimate the probability density function and the probability measure function of the BGPD model in Theorem~\ref{theorem:logistic}. Finally, Theorems~\ref{theorem:validitylogistic} and \ref{theorem:validitypoisson} express the requirements of a valid BGPD model based on the mean constraint on the probability measure functions of the Pickands coordinates corresponding to the BGPD models obtained in Theorems~\ref{theorem:logistic} and \ref{theorem:pointprocess}, respectively.

\begin{definition}[\textbf{Fr\'echet transformation}]
\label{def:frechet}
Suppose $(x_1,y_1),...,(x_n,y_n)$ are independent realizations of a random variable $(X,Y)$ with joint distribution $F(x,y)$. For optimum thresholds $u_{x}$ and $u_{y}$, each marginal distribution of $F$ has an approximation of the form (\ref{eqn:gpd_dist}), with the parameters $(\tilde{\sigma}_x,\xi_x)$ and $(\tilde{\sigma}_y,\xi_y)$ for $X$ and $Y$, respectively.
Let's apply Fr\'echet transformation to variables $X$ and $Y$ to induce new variables $\tilde{X}$ and $\tilde{Y}$, given by

\begin{equation}
\label{eqn:x_tilde}
    \tilde{X} = -\Big( \log \Big\{ 1-\zeta_{x} \Big[1+\frac{\xi_x (u_x-X)}{\tilde{\sigma}_{x}}\Big]^{-1/\xi_x} \Big\} \Big) ^{-1}, \hspace{.3cm} X<u_x,
\end{equation}
and
\begin{equation}
\label{eqn:y_tilde} 
    \tilde{Y} = -\Big( \log \Big\{ 1-\zeta_{y} \Big[1+\frac{\xi_y (u_y-Y)}{\tilde{\sigma}_{y}}\Big]^{-1/\xi_y} \Big\} \Big) ^{-1}, \hspace{.5cm} Y<u_y.
\end{equation}
where $\zeta_{x} = Pr(X<u_x)$ and $\zeta_{y} = Pr(Y<u_y)$.
Then, the joint distribution function of $\tilde{x}$ and $\tilde{y}$, $\tilde{F}$, has margins that are approximately standard Fr\'echet distribution, i.e., $\tilde{F}(\tilde{x}) = exp(-1/\tilde{x}), \tilde{x}>0$, $\tilde{F}(\tilde{y}) = exp(-1/\tilde{y}), \tilde{y}>0$, where $\tilde{x}$ and $\tilde{y}$ are any realizations from $\tilde{X}$ and $\tilde{Y}$, respectively \cite{evt_04}.
\end{definition}

\begin{theorem}
\label{theorem:bevt}
Let us define $M^{*}_n = (M^{*}_{\tilde{x},n}, M^{*}_{\tilde{y},n})$, where $M^{*}_{\tilde{x},n} = \max_{\substack{i=1,...,n}} \{\tilde{X}_i\}/n$, and $M^{*}_{\tilde{y},n} = \max_{\substack{i=1,...,n}} \{\tilde{Y}_i\}/n$, and $(\tilde{X}_i, \tilde{Y}_i)$ are independent vectors with standard Fr\'echet marginal distributions as expressed in (\ref{eqn:x_tilde}) and (\ref{eqn:y_tilde}), respectively, and $n$ is the number of realizations in the tail of $X$ or $Y$. Then, 
\begin{equation}
    \label{eqn:bevt_1}
    Pr\{M^{*}_{\tilde{x},n} \leq \tilde{x}, M^{*}_{\tilde{y},n} \leq \tilde{y}\} \xrightarrow{d} G(\tilde{x},\tilde{y}),
\end{equation}
where $\xrightarrow{d}$ denotes limit of distribution, and $G$ is a bi-variate non-degenerate distribution function with the form 
\begin{equation}\label{eqn:nondegenerate}
    G(\tilde{x},\tilde{y}) = \exp\{ -V(\tilde{x},\tilde{y}) \}, \hspace{0.7cm} \tilde{x}>0,\tilde{y}>0 
\end{equation}
 
\begin{equation}\label{eqn:theorem_V}
    V(\tilde{x},\tilde{y}) = 2 \int \limits_0^1 \max (\frac{\omega}{\tilde{x}},\frac{1-\omega}{\tilde{y}}) H(d\omega),
\end{equation}
and $H$ is a distribution function on $[0,1]$ satisfying the mean constraint 
\begin{equation}\label{eqn:theorem_H}
    \int \limits_0^1 \omega H(d\omega) = 1/2.
\end{equation}

Any family of distributions that arise as limits in (\ref{eqn:bevt_1}) can be considered the class of bi-variate extreme value distributions \cite{evt_04}, \cite{bgpd_01}.
\end{theorem}

\begin{proof}
Please refer to \cite{mevt_10} for the proof.
\end{proof}

Theorem~\ref{theorem:bevt} implies that the class of bi-variate extreme value distributions is in one-to-one correspondence with the set of distribution functions $H$ on $[0, 1]$ satisfying (\ref{eqn:theorem_H}). If the probability measure $H$ is differentiable with probability density $h$, integral (\ref{eqn:theorem_V}) is simplified to 
\begin{equation}\label{eqn:theorem_Vrev}
    V(\tilde{x},\tilde{y}) = 2 \int_0^1 \max (\frac{\omega}{\tilde{x}},\frac{1-\omega}{\tilde{y}}) h(\omega)d\omega.
\end{equation}

This assumes that the thresholds $u_x$ and $u_y$ are small enough to justify the limit (\ref{eqn:bevt_1}) as an approximation \cite{evt_04}. 

Theorem~\ref{theorem:bevt} also implies that for generating a class of BGPD models expressed in (\ref{eqn:nondegenerate}), it is required to obtain a parametric family for $H$ over $[0,1]$  whose mean is $0.5$ for every value of the parameter. However, in practice, it is hard to find parametric families whose mean is parameter-free and for which the integral (\ref{eqn:theorem_Vrev}) is tractable. Two approaches that can be utilized to model the bi-variate extreme value distribution are the logistic family and the Poisson point process, which will be discussed in the following.

\subsubsection{BGPD Based on Logistic Family}
\label{sec:backgroundlogistic}

\begin{theorem}
\label{theorem:logistic}
Let $(\tilde{X}, \tilde{Y})$ be independent vectors with standard Fr\'echet marginal distributions as expressed in (\ref{eqn:x_tilde}) and (\ref{eqn:y_tilde}). The logistic family $G(\tilde{x},\tilde{y})$ is a standard class expressing BGPD of $\tilde{X}$ and $\tilde{Y}$ as:
\begin{equation}
    \label{eqn:Glogistic}
    G_{l}(\tilde{x},\tilde{y}) = \exp{\big\{ -V(\tilde{x},\tilde{y})}\big\},
\end{equation}
where 
\begin{equation}
    \label{eqn:Vlogistic}
    V(\tilde{x},\tilde{y}) = \big( \tilde{x}^{-1/\alpha} + \tilde{y}^{-1/\alpha} \big)^{\alpha},
\end{equation}
$\alpha \in (0,1)$ denotes the dependency factor between variables $\tilde{x}$ and $\tilde{y}$ with the constraint $\tilde{x},\tilde{y}>0$, and $\tilde{x}$ and $\tilde{y}$ are expressed as (\ref{eqn:x_tilde}) and (\ref{eqn:y_tilde}), respectively, for large enough thresholds $u_x$ and $u_y$ \cite{evt_04}, \cite{bgpd_01}. 
\end{theorem}

\begin{proof}
Please refer to \cite{evt_04}, \cite{evt_03}, \cite{mevt_10} for the proof.
\end{proof}

The variables $\tilde{x}$ and $\tilde{y}$ of the logistic distribution family in (\ref{eqn:Glogistic}) are exchangeable and have correlation $\rho = 1- \alpha^2$ \cite{evt_03}. The $\alpha$ value very close to $1$ denotes strong dependence between the variables $x$ and $y$, even at moderately extreme levels.

\subsubsection{BGPD Based on Poisson Point Process}
\label{sec:backgroundpointprocess}
\begin{definition}[\textbf{Pickands coordinates}]
\label{def:pickand}
Let $(\tilde{X}, \tilde{Y})$ be independent vectors with standard Fr\'echet marginal distributions as expressed in (\ref{eqn:x_tilde}) and (\ref{eqn:y_tilde}). Let's define the transformation $T_{P}(\tilde{x},\tilde{y})$ by
\begin{equation*}
    T_{P}(\tilde{x},\tilde{y}) := \big( \frac{-\tilde{x}}{n}+\frac{-\tilde{y}}{n}, \frac{-\tilde{x}/n}{-\tilde{x}/n-\tilde{y}/n} \big) =: (\omega,r).
\end{equation*}
where $n \in N$, is the length of vector $\tilde{X}$ or $\tilde{Y}$.
$T_P(\tilde{x},\tilde{y})$ is called transformation with respect to the Pickands coordinates $(\omega,r)$, where $\omega$ is pseudo-polar angular component measuring angle in $[0,1]$ scale, and $r$ is pseudo-polar radial component measuring the distance from the origin \cite{evt_04}. This transformation is one-to-one with the inverse
\begin{equation*}
    T^{-1}_P(\omega,r) =: r\big( \omega,1-\omega \big) =: (\tilde{x},\tilde{y}).
\end{equation*}

The pseudo-polar Pickands mapping has the same geometrical interpretation as standard polar coordinates with the difference that polar coordinates use the Euclidian norm for the angular and radial component, while the Pickands coordinates use the sum norm \cite{mevt_10}-\cite{mevt_12}.
\end{definition}

\begin{theorem}
\label{theorem:pointprocess}
Assume that $(\tilde{X},\tilde{Y})$ be a vector of independent bi-variate observations with standard Fr\'echet margins satisfying (\ref{eqn:bevt_1}). Let $N_n$ denote a sequence of point processes defined as
\begin{equation}
\label{eqn:Nn}
    N_n = \{(n^{-1}\tilde{x}_1,n^{-1}\tilde{y}_1),...,(n^{-1}\tilde{x}_n,n^{-1}\tilde{y}_n)\}.
\end{equation}
where $\tilde{x}_i$ and $\tilde{y}_i$ are the $i^{th}$ realization of $\tilde{X}$ and $\tilde{Y}$, respectively. Then, $N_n$ converges to $N$, i.e., $N_n \xrightarrow{d} N$, where $N$ is a non-homogeneous Poisson process on space $A$ defined by
\begin{equation*}
    A = \{(0,\infty) \times (0,\infty) \backslash (0,\tilde{x}) \times (0,\tilde{y}) \}, 
\end{equation*}
denoting space $\{(0,\infty) \times (0,\infty)\}$ excluding sub-space $\{(0,\tilde{x}) \times (0,\tilde{y})\}$. Therefore, according to Poisson point process limit, the bi-variate extremes are modeled as 
\begin{equation}
\label{eqn:GPDpoisson}
    G_{pp}(\tilde{x},\tilde{y}) \rightarrow Pr\{N(A) = 0\} = \exp\{-\Lambda(A)\},
\end{equation}
where the intensity measure of Poisson point process, $\Lambda(A)$, given by \cite{evt_04}, \cite{mevt_12}
\begin{equation}
\label{eqn:capitalLambdaThm}
    \Lambda(A)= -2 \int_{\omega=0}^{1} \max \big(\frac{\omega}{-\tilde{x}/n},\frac{1-\omega}{-\tilde{y}/n}\big) \,H_{pp}(d\omega),
\end{equation}
has the same form as function $V(\tilde{x},\tilde{y})$ in (\ref{eqn:theorem_V}), and the probability measure function $H_{pp}(d\omega)=\int_0^1 h_{pp}(\omega)\,d\omega$ with the probability density function $h_{pp}(\omega)$ of the angular component $\omega$ based on Definition~\ref{def:pickand} \cite{evt_04}, \cite{mevt_11}, \cite{mevt_12}. 
\end{theorem}

\begin{proof}
Let $r$ and $\omega$ denote the Pickands coordinates of $\tilde{x}$ and $\tilde{y}$ given by
\begin{equation}
    r = \frac{-\tilde{x}}{n}+\frac{-\tilde{y}}{n}, \hspace{1cm} and \hspace{1cm} \omega = \frac{-\tilde{x}/n}{r},
\end{equation}
where $\tilde{x}$ and $\tilde{y}$ have standard marginal Fr\'echet distributions. Then, the intensity function of $N$ on space $A$ in Theorem~\ref{theorem:pointprocess} is \cite{evt_04}, \cite{mevt_11}, \cite{mevt_12}
\begin{equation}
\label{eqn:lambda}
    \lambda(r,w) = 2\frac{dr}{r^2}\,H_{pp}(d\omega).
\end{equation}
Additionally, let $N_n$ be the point process defined as (\ref{eqn:Nn}) with intensity function $\lambda(r,\omega)$ in (\ref{eqn:lambda}). Then, $\Lambda(A)$ in Theorem~\ref{theorem:pointprocess} is given by
\begin{align*}
\begin{split}
    \Lambda(A) = \int_{A} 2\frac{dr}{r^2} \,H_{pp}(d\omega)\
    = \int_{\omega=0}^{1} \int_{-\infty}^{r_{max}} 2\frac{dr}{r^2} \,H_{pp}(d\omega),
\end{split}
\end{align*}
where $r_{max}$ is the maximum of $\frac{-\tilde{x}/n}{\omega}$ and $\frac{-\tilde{y}/n}{1-\omega}$, depending on how $r$ is defined based on $\omega=\frac{-\tilde{x}/n}{r}$ or $1-\omega= \frac{-\tilde{y}/n}{r}$. Therefore,
\begin{align}
\label{eqn:FinalcapitalLambda}
\begin{split}
    \Lambda(A)
    = \int_{\omega=0}^{1} \int_{-\infty}^{r=\max\{ \frac{-\tilde{x}/n}{\omega},\frac{-\tilde{y}/n}{1-\omega}\}} 2\,\frac{dr}{r^2} \,H_{pp}(d\omega)\\
    = -2 \int_{\omega=0}^{1} \max \big(\frac{\omega}{-\tilde{x}/n},\frac{1-\omega}{-\tilde{y}/n}\big) \,H_{pp}(d\omega).
\end{split}
\end{align}

\end{proof}

\begin{remark}
\label{remark:meanconstraint}
The Pickands probability measure $H_{pp}(\omega)$ has the following property \cite{mevt_12}- \nocite{pickands_01}\nocite{pickands_02}\cite{pickands_03}: %
\begin{equation}
\label{eqn:pickandsdensityvalid}
    \int_{[0, 1]} \omega\,H_{pp}(d\omega) =  \int_{[0, 1]} (1-\omega)\,H_{pp}(d\omega).
\end{equation}
\end{remark}

The fact expressed in Remark~\ref{remark:meanconstraint} will be used in Theorem~\ref{theorem:validitypoisson} to assess the validity of the Poisson point process-based BGPD.

\subsubsection{BGPD Model Assessment}
\label{sec:background_assessment}
Any distribution function $H$ defined on the space $[0, 1]$ in (\ref{eqn:theorem_V}) that satisfies the mean constraint in (\ref{eqn:theorem_H}), gives rise to a valid limit in (\ref{eqn:bevt_1}) \cite{evt_04}. Therefore, if $G(\tilde{x},\tilde{y})$ is a valid model to estimate the tail of bi-variate extremes, its corresponding probability measure function $H(\omega)$ should satisfy the $0.5$ mean constraint according to Theorem~\ref{theorem:bevt}.

\begin{definition}[\textbf{Pickands constraints}]
\label{def:pickandsconstraints}
Let $(\tilde{X}, \tilde{Y})$ be independent vectors with standard Fr\'echet marginal distributions as expressed in (\ref{eqn:x_tilde}) and (\ref{eqn:y_tilde}), and corresponding Pickands coordinates $(r,\omega)$. Then, for a radial cut-off point $r_0 < 0$ close to $0$, we have the following Pickands constraints \cite{mevt_10}:
\begin{enumerate}
    \item Conditional on $r>r_0$, the radial and angular components of the Pickands coordinates ($\omega,r$), are independent. 
    \item Conditional on $r>r_0$, the radial component of the Pickands coordinates is uniformly distributed. Therefore,
    \begin{equation}
        \label{eqn:densityC} 
        Pr(r \ge R | r>r_0) = \frac{R}{r_0}, \hspace{1cm} r_0 \leq R \leq 0.
    \end{equation}
\end{enumerate}
\end{definition}

\begin{theorem}
\label{theorem:validitylogistic}
If the BGPD model based on the logistic distribution family $G_{l}(\tilde{x},\tilde{y})$, given by (\ref{eqn:Glogistic}), is a valid model to estimate the tail of the bi-variate extremes, the mean constraint defined in (\ref{eqn:theorem_H}) should be satisfied on $H_{l}(\omega)=\int h_{l}(\omega)\,d\omega$, where
\begin{equation}
\label{eqn:densityZ}
    h_{l}(\omega) = \frac{\phi(\omega)}{\mu},
\end{equation}
$\mu = \int\phi(\omega)\mathrm{d}\omega >0$, and $\phi(\omega)$ is the Pickands density function given by
\begin{equation}
\label{eqn:phi_omega_r}
    \phi(\omega) = |r| \Big( \frac{\partial ^2 }{\partial \tilde{x} \partial \tilde{y}} G_{l} \Big) \big(T_P^{-1}(\omega,r)\big), \hspace{1cm} r>r_0,
\end{equation}
where $G_{l}$ is the BGPD expressed as (\ref{eqn:Glogistic}) and $r_0$ is the optimum cut-off point obtained based on the constraints on the Pickands coordinates in Definition~\ref{def:pickandsconstraints}.
\end{theorem}

\begin{proof}
Please refer to \cite{mevt_10} and \cite{mevt_12} for the proof.
\end{proof}

\begin{theorem}
\label{theorem:validitypoisson}
If the BGPD model $G_{pp}(\tilde{x},\tilde{y})$ based on the Poisson point process approach given by (\ref{eqn:GPDpoisson}) is a valid model to estimate the tail of the bi-variate extremes, the mean constraint defined in (\ref{eqn:theorem_H}) is always satisfied for the $H_{pp}(\omega)$ function in (\ref{eqn:capitalLambdaThm}).
\end{theorem}

\begin{proof}
Let $r$ and $\omega$ denote the Pickands coordinates of $\tilde{x}$ and $\tilde{y}$ based on Definition~\ref{def:pickand}. Then, referring to (\ref{eqn:pickandsdensityvalid}), we have 
\begin{equation*}
    2\int_{[0,1]} \omega\,H_{pp}(d\omega) = \int_{[0,1]}H_{pp}(d\omega),
\end{equation*}
where $\int_{[0,1]}H_{pp}(d\omega) = 1$, and therefore, the $0.5$ mean constraint defined in (\ref{eqn:theorem_H}) is satisfied for the probability measure function $H_{pp}(\omega)$.
\end{proof}

\section{Proposed Channel Modeling Methodology}
\label{sec:methodology}
We propose a novel BEVT-based channel modeling methodology with the goal of estimating the joint lower tail statistics of multiple channels in the ultra-reliable regime.
The methodology consists of the following steps: The sequences of measured received power samples are converted into sequences of i.i.d. samples by removing their dependency via declustering, where the samples are divided into multiple clusters, each of which includes consecutive dependent observations, and clusters are separated by a specific sample gap to ensure the independency between clusters \cite{MehrniaTWC}. Upon applying EVT to the resulting sequence of i.i.d. samples, the optimum thresholds are determined.
Since the bi-variate analysis is restricted to the time intervals in which both thresholds are exceeded, we keep only the exceedances that occur in the same time interval. Otherwise, we ignore the exceedance in any sequence. Thereafter, the parameters of the UGPD associated with the optimum thresholds are estimated by using the MLE. Then, we assess the validity of the fitted UGPD model to the tail distribution by using the probability plots, including the probability/probability (PP) plot and the quantile/quantile (QQ) plot. Afterward, we model the inter-relationship of bi-variate extremes based on Theorem~\ref{theorem:bevt} by applying two methods, logistic distribution and Poisson point process. In the logistic distribution approach, upon determining the dependency factor $\alpha$ between the Fr\'echet transformed variables in Definition~\ref{def:frechet}, we model the tail of the joint PDF based on Theorem~\ref{theorem:logistic} and verify the model according to Theorem~\ref{theorem:validitylogistic} by using the doubled-transformed data to the Pickands coordinate based on Definition~\ref{def:pickand}. In the Poisson point process approach, based on Theorem~\ref{theorem:pointprocess}, we transform the data sequence in two steps: First, based on the Fr\'echet transformation in Definition~\ref{def:frechet}, and second, based on Pickands coordinates in Definition~\ref{def:pickand}. Then, we determine the probability measure of the Pickands angular coordinate $H_{pp}(\omega)$ and the intensity measure of Poisson point process $\Lambda(A)$. Finally, we assess the validity of the Poisson point process-based bi-variate model based on Theorem~\ref{theorem:validitypoisson}.
The proposed algorithm for the bi-variate extremes is depicted in Fig.~\ref{fig:maindiagram} and explained in detail next.

\begin{figure*}[ht]
    \centering
    \includegraphics[width=\linewidth]{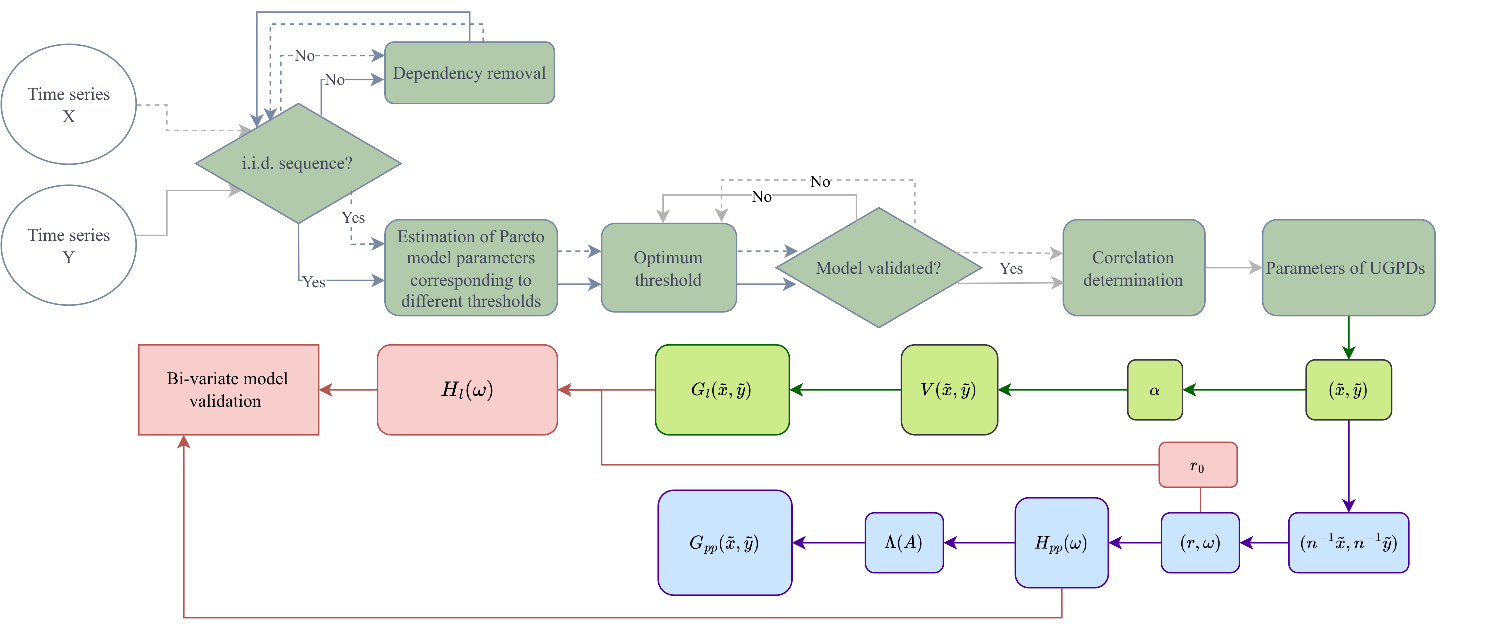}
    \caption{Flowchart of the proposed multi-channel modeling framework.}
    \label{fig:maindiagram}
\end{figure*}
\subsection{Modeling the Uni-variate Extremes}
\label{sec:methodologyUGPD}
\subsubsection{Declustering Approach}
In the declustering approach, first, we assume a low enough threshold $u$ and start the first cluster by the first sample $x_i$ below this threshold. Then, all the successive samples below threshold $u$ followed by the first sample of the initial cluster are assigned to the first cluster. Right after observing a sample over $u$, we let the cluster continue for $mg$ more consecutive values and then, if no value below $u$ is detected, close the cluster. The next cluster starts with the following value below the threshold $u$. Upon determination of the clusters of samples, we extract the minimum value of each cluster, apply EVT to the i.i.d. cluster minima, and model their tail distribution by using GPD. The declustering approach is based on the fact that the minimum values of the clusters are far enough to be considered independent and identically distributed \cite{MehrniaTWC}, \cite{evt_04}.

\subsubsection{Optimum Threshold Determination}
Determination of the optimum threshold is of paramount importance as it specifies where the tail of the distribution starts and distinguishes between extreme events happening rarely and non-extreme values. The optimum threshold of each uni-variate case is determined by using the mean residual life (MRL) and parameter stability methods. Based on the MRL method, $u$ is the optimum threshold if $u$ is the highest threshold below which the expected value of samples exceeding $u$ is a linear function of $u$, i.e., $E(u-X|X<u)$ is linear against $u$. The remarkable advantage of the MRL method is its simplicity, as it can be applied to the data sequence prior to the estimation of the UGPD parameters. However, due to the lower precision of the MRL method compared to the parameter stability method, it is sometimes difficult to obtain the optimum threshold explicitly. Therefore, the MRL method is usually utilized as a complementary method or in the case that the optimum threshold determination does not require high precision.
The parameter stability method states that the optimum threshold $u$ is the highest threshold below which the estimated shape and modified scale parameters of the UGPD fitted to the tail distribution associated with a variety of thresholds is a linear function of the threshold. It is worth noting that the modified scale parameter is defined as $\sigma^{*} = \tilde\sigma_{u} - \xi u$. Additionally, the linearity relation is assessed by using the R-squared statistical measure, denoted by $R^{2}$ \cite{MehrniaTWC}-\cite{MehrniaTVT}.

\subsubsection{Model Validity Assessment}
The validity of the GPD model is assessed by using probability plots, i.e., probability/probability (PP) and quantile/quantile (QQ) plots. These plots are graphical techniques used to assess the validity of the models fitted to the empirical values. In the PP plot, we plot the empirical CDF of the occurrence of an extreme value versus the corresponding CDF obtained by the GPD, while in the QQ plot, we plot the empirical extreme quantile versus the corresponding value obtained by the inverse of GPD \cite{MehrniaTWC}. If the GPD appropriately models the extreme values exceeding threshold $u$, then both PP and QQ plots should fit the unit diagonal line, i.e., the $45^\circ$ line \cite{evt_04}, \cite{evt_01}.

\subsection{Evaluation of the Correlation Coefficient}
The first step in the bi-variate extreme analysis is to assess the amount of dependence in the tails \cite{mevt_11}. 
However, since the bi-variate analysis is restricted to those time intervals in which both thresholds $u_x$ and $u_y$ are exceeded, before determining the correlation among the tail samples, we are required to revise the tail samples by removing those exceedances happening in one sequence but not in the other one, at the same time interval \cite{evt_04}. Accordingly, we consider a time interval consisting of $M$ samples and check if $u_x$ and $u_y$ are exceeded simultaneously in sequences $X$ and $Y$, respectively. If so, we capture the minima of clusters; otherwise, we ignore the minima of each cluster \cite{evt_04}, \cite{mevt_11}. Then, EVT is applied to the obtained exceedances to determine the UGPD parameters.

Upon determining the extremes of each sample sequence, if the extremes of the two sequences are independent, no bi-variate modeling like (\ref{eqn:bevt_1}) is required. Otherwise, the existence of correlation among the tail samples suggests the requirement of investigating the inter-relationship of the extreme values of receivers.
On the other hand, correlation among the received powers in the total samples is used to assess the feasibility of the spatial diversity. If the correlation coefficient among the total samples is between $0.1$ and $0.5$, the spatial diversity is suggested \cite{corcoef_01}. Otherwise, the employment of spatial diversity is pointless due to the concurrent fading at different links. If the spatial diversity is feasible and the correlation among the tail samples is high enough, we are required to analyze the bi-variate tail characteristics for a URC system.

\subsection{Modeling the Multivariate Extremes Based on the Logistic Distribution Approach}
\label{sec:BGPD_logistic}
\subsubsection{Data Conversion Based on the Fr\'echet Transformation} 
Applying the Fr\'echet transformation by using the Definition~\ref{def:frechet} on the obtained tail sequences $X$ and $Y$ from Section~\ref{sec:methodologyUGPD}, we induce variables $\tilde{X}$ and $\tilde{Y}$ whose marginal distribution functions have Fr\'echet distribution for $x < u_x$ and $y< u_y$, approximately. This transformation is required as the input of the bi-variate extreme modeling approach is expected to have the Fr\'echet distribution \cite{evt_04}, \cite{evt_03}.  

\subsubsection{Determining the Dependency Factor $\alpha$}
To determine the dependency parameter $\alpha$ of the logistic model in (\ref{eqn:Vlogistic}), we either estimate the correlation coefficient $\rho$ between the variables $\tilde{x}$ and $\tilde{y}$ and then, calculate $\alpha$ as $\sqrt{1-\rho}$, or directly estimate it by MLE, where the log-likelihood function is defined as 
\begin{equation}
    l(\tilde{x},\tilde{y}|\alpha) = \sum_{i=1}^{n} \ln \, v_{\alpha}(\tilde{x}_i,\tilde{y}_i), \end{equation}
    where,
\begin{equation}
        v_{\alpha}(\tilde{x},\tilde{y}) = \frac{\partial ^2}{\partial \tilde{x} \partial \tilde{y}} V(\tilde{x},\tilde{y}),
\end{equation}   
and $V(\tilde{x},\tilde{y})$ is the logistic distribution function defined in (\ref{eqn:Vlogistic}).

\subsubsection{BGPD Model Based on the Logistic Distribution}
The bi-logistic family of distributions formulated in (\ref{eqn:Vlogistic}) is used to model the bi-variate extremes based on Theorem~\ref{theorem:bevt}. Upon obtaining the dependency parameter $\alpha$, as well as the Fr\'echet transformed variables $\tilde{x}$ and $\tilde{y}$, we build function (\ref{eqn:Vlogistic}) and then, by taking the exponential of its negative function, we determine the BGPD as expressed in (\ref{eqn:Glogistic}). If this BGPD model is reliable to characterize the tail of the bi-variate distribution, the mean constraint on the probability measure function defined in Theorem~\ref{theorem:validitylogistic} should be satisfied.

\subsection{Modeling the Multivariate Extremes Based on the Poisson Point Process Approach}
\label{sec:bgpd_pointprocess}
\subsubsection{Pseudo-polar Pickands Coordinates Transformation}
In Poisson point process-based approach for modeling the bi-variate extremes, it is required first to transform the data from Cartesian to the polar coordinate: $(\tilde{x},\tilde{y})\rightarrow(r,\omega)$. To this end, we apply pseudo-polar Pickands transformation based on Definition~\ref{def:pickand} that induces two new variables, radial component $r$ and angular component $\omega$.

\subsubsection{Determining the Probability Measure Function of the Radial and Angular Components}
The CDF of the probability measure function of the angular component $H_{pp}(\omega)$ is determined based on Theorem~\ref{theorem:pointprocess}. This probability measure will be used in the next step to determine the intensity of Poisson point process $\Lambda(A)$ and, later, to assess the validity of the proposed Poisson point process-based BGPD, according to Theorem~\ref{theorem:validitypoisson}.

\subsubsection{BGPD Model Based on the Poisson Point Process}
Upon determining the angular component probability measure function $H_{pp}(\omega)$, we compute the density function of the Poisson point process $\Lambda(A)$ for the defined space $A$ based on Theorem~\ref{theorem:pointprocess}. $\Lambda(A)$ is actually equivalent to $V(\tilde{x},\tilde{y})$ in the logistic distribution based-BGPD model. Afterward, the BGPD model based on the Poisson point process approach is determined as $\exp(-\Lambda(A))$.

\subsection{BGPD Model Validation}
\label{sec:bgpd_modelvalidity}
\subsubsection{Model Validation for Logistic-based BGPD}
According to Theorem~\ref{theorem:validitylogistic}, if the logistic-based BGPD is a valid model to estimate the tail of the bi-variate extremes, the probability measure function $H_{l}(\omega)$, conditional on $\{\omega: r>r_0\}$, needs to satisfy the $0.5$ mean constraint, i.e., $\int_{\omega} \omega \, H_{l}(d\omega) = 0.5$, where $\omega$ is the Pickands angular component of $\tilde{x}$ and $\tilde{y}$.

To determine $r_0$ based on Definition~\ref{def:pickandsconstraints}, a plot of $r$ versus $\omega$ for $r<r_0$, referred to as the $r$-$\omega$ plot, can be utilized to address the first constraint in Definition~\ref{def:pickandsconstraints} by checking the dependency between $r$ and $\omega$. The dependency of $r$ and $\omega$ is assessed based on the correlation results in which $r$ and $\omega$ are independent if their correlation is less than the critical value $0.05$. Additionally, to address the second constraint in Definition~\ref{def:pickandsconstraints}, the distribution of the radial components $r_0<r<0$ is expected to be fitted to the uniform distribution, where $Pr(r>R|r>r_0)=\frac{R}{r_0}$.

\subsubsection{Model Validation for Poisson Point Process-based BGPD}
According to Theorem~\ref{theorem:validitypoisson}, the Poisson point process-based BGPD is, by default, a valid model to estimate the tail of the bi-variate extremes as its intensity function $H_{pp}(\omega)$ satisfies the mean constraint. Therefore, the constraint $\int_{\omega} \omega \, H_{pp}(d\omega) = 0.5$ is insured in Poisson point process-based BGPD while determining $H_{pp}(\omega)$.

The complexity of the proposed MEVT-based channel modeling framework is $O(n\,N_{Tx}\,N_{Rx})$, where $N_{Tx}$ and $N_{Rx}$ are the number of transmitters and receivers, respectively, and $n$ is the number of the training samples for individual channel sequences.

\section{Numerical Results}
\label{sec:numerical_results}
The goal of this section is to evaluate the performance of the proposed channel modeling algorithm in estimating the BGPD model fitted to the joint distribution of the channel data sequences based on two approaches: the logistic distribution approach and  the Poisson point process approach, and also compare their performances with the traditional extrapolation-based approaches to estimate the statistics of the channel tail for a system operating in the spatial diversity in MIMO-URC. In the traditional extrapolation-based approach, upon estimating the distribution of the existing channel data for the reliability order of $10^{-3}$-$10^{0}$ PER \cite{vehicular_01}-\cite{vehicular_02} for individual channel data sequences, we determine their joint probability distribution, and then, extrapolate it towards the ultra-reliable region $10^{-9}$-$10^{-5}$ PER \cite{urllc_02}. 

We have collected channel measurement data within the engine compartment of Fiat Linea under various engine and driving scenarios at $60$ Gigahertz (GHz) by using a Vector Network Analyzer (VNA) (R$\And$S$\textsuperscript{\textregistered}$ ZVA$67$). The VNA is connected to the transmitter and receivers through the R$\And$S$\textsuperscript{\textregistered}$ ZV-Z$196$ and PE$361$ port cables, respectively, with $610$ millimeter (mm) length as shown in Fig~\ref{fig:fiatlinea}. The output power range at the port can be decreased to $-100$~dBm to measure deep fading. The transmitter is an omnidirectional antenna operating from $58$~GHz to $63$~GHz with $0$ decibel isotropic (dBi) nominal gain. The receivers are horn antennas operating between $50$-$75$ GHz with a nominal $24$ dBi gain and $11^\circ$ and $9.5^\circ$ horizontal and vertical half power beamwidth, respectively. The antennas are connected to the coax cables through the waveguide to the coax adaptor operating at the frequency span of $50-65$~GHz, with insertion loss $0.5$ decibel (dB) and impedance $50$ Ohm ($\Omega$).

The data were collected while driving the car at the Koc University campus and emulating different scenarios, including starting/stopping the car, moving up/down on a ramp, and driving on a flat road. More than $10^{6}$ successive samples are captured from each receiver antenna for about $1.5$~hours with a time resolution of $3$ ms. We use \textsc{MATLAB} for analyzing the data and the implementation of the proposed algorithm as well as the traditional extrapolation-based approach \cite{urllc_02}. Meanwhile, the estimation error of the proposed methodology and the extrapolated-based approach have been reported by means of the Root Mean Square Error (RMSE) metric.

In the following, first, in Section~\ref{sec:performance_UGPD}, we provide the numerical results in the determination of the optimum threshold over which the tail statistics are derived for two uni-variate channel data sequences obtained from receivers Rx$1$ and Rx$2$, and then, validate the tail model corresponding to the optimum threshold by means of the probability plots. Upon determining the correlation among the tail samples in Section~\ref{sec:performance_corr}, we apply the logistic distribution approach to model the tail distribution of the bi-variate extremes in Section~\ref{sec:performance_BGPD_logistic}. Next, in Section~\ref{sec:performance_BGPD_pointprocess}, we estimate the tail distribution of the bi-variate extremes by applying the Poisson point process approach. Afterward, in Section~\ref{sec:performanceValidation}, we assess the validity of the determined BGPD models based on logistic distribution and Poisson point process approaches. Finally, we compare the performance of both proposed methodologies, logistic distribution-based and Poisson point process-based approaches, to that of the traditional extrapolation-based technique in the estimation of the bi-variate channel tail statistics in  Section~\ref{sec:performanceComparison}.

\begin{figure}[ht]
    \centering
    \includegraphics[width=0.8\columnwidth]{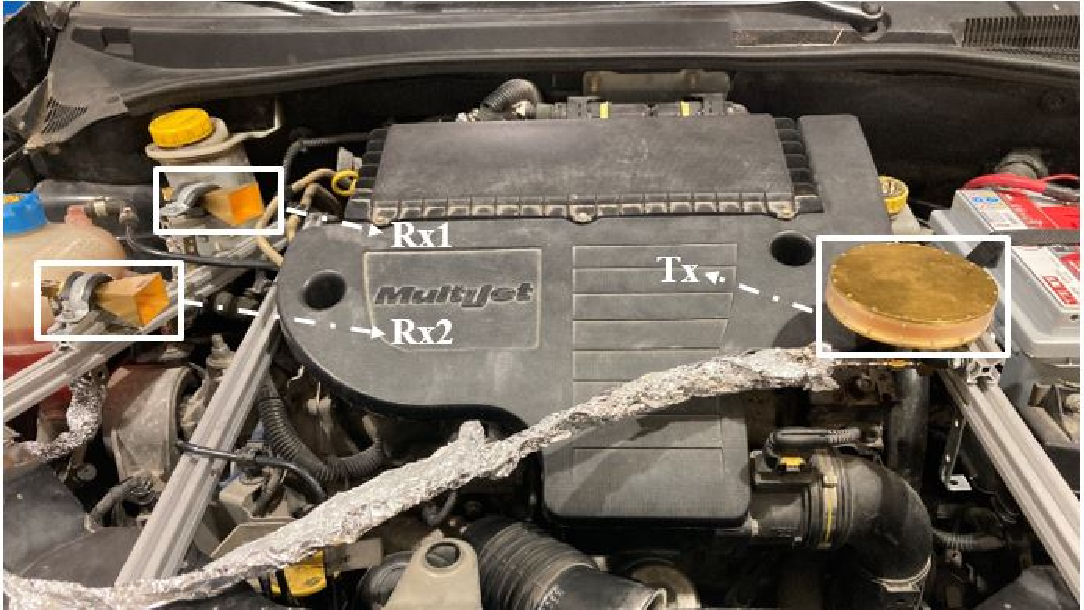}
    \caption{Transmitter and receiver antennas in the engine compartment.}
    \label{fig:fiatlinea}
\end{figure}
\vspace{-0.3cm}
\subsection{UGPD Model}
\label{sec:performance_UGPD}

Towards determining the optimum threshold of the GPD fitted to the tail distribution of samples in each group, the declustering approach is used to remove the dependency among the samples and obtain i.i.d. observation for the EVT input. In this regard, the MRL plot and parameter stability method are utilized to determine the thresholds $u_x$ and $u_y$, and the minimum gaps(mgs) between the samples mg$_1$ and mg$_2$, for receiver Rx$1$ and Rx$2$, respectively. We skip the illustration of these results and only present the critical information as the subject of the uni-variate channel tail modeling has been discussed in our previous paper \cite{MehrniaTWC} in detail.
According to the results, for the uni-variate channel data sequence obtained from receiver Rx$1$, $u_x = -15$~dBm is the optimum threshold below which the $R^{2}>0.95$ for all mg$_1>1$. Additionally, $u_y = -30$~dBm is the optimum threshold below which the linearity condition is observed for all mg$_2>1$ for the uni-variate channel data of receiver Rx$2$.
Upon determining the optimum thresholds ($u_x$ and $u_y$) and minimum gaps mg$_1$ and mg$_2$ between the clusters, the probability plots, including the PP plot and QQ plot, are used to validate the accuracy of the UGPD models fitted to the channel tail distribution of the i.i.d. samples obtained from receivers Rx$1$ and Rx$2$, corresponding to the optimum thresholds $u_x$ and $u_y$, respectively.

\subsection{Determination of Correlation}
\label{sec:performance_corr}
To determine the correlation among the tail samples, first, we have considered the time intervals consisting of $1000$ samples, and then, at each interval, obtained those minima of Rx$1$ and Rx$2$ sequences that simultaneously exceed $u_x=-15$~dBm and $u_y=-30$~dBm, respectively. The parameters of UGPD fitted to these new exceedances are $(\xi_x, \tilde{\sigma}_x) = (-0.1469,4.0367)$, and $(\xi_y,\tilde{\sigma}_y) = (-0.4245,8.5886)$ for receivers Rx$1$ and Rx$2$, respectively.
The correlation among the total samples obtained from Rx$1$ and Rx$2$ is about $0.2766$, which confirms the applicability of spatial diversity as it is less than $0.5$. Additionally, the $0.3957$ correlation coefficient among the tail samples indicates the existence of correlation in the tail samples of Rx$1$ and Rx$2$ and, therefore, confirms the necessity of studying the inter-relationships of the extreme values.

\subsection{BGPD Based on the Logistic Distribution Approach}
\label{sec:performance_BGPD_logistic}
Fig.~\ref{fig:frechet} illustrates the CDF of the transformed data by using the Fr\'echet transformation for the received power tail samples obtained from receivers Rx$1$ and Rx$2$. 
For small values of $\Tilde{x}$ and $\Tilde{y}$, the Fr\'echet transformation is not able to fit appropriately to the empirical data. However, as the transformed variables $\Tilde{x}$ and $\Tilde{y}$ increase and correspond to the extreme values with extreme fading, the Fre\'chet distribution fits well to the empirical values. This confirms that $\tilde{x}$ and $\tilde{y}$ are the reliable transformed
variables for building the BGPD model.

\begin{figure}[ht]
\centering
\captionsetup[subfigure]{labelformat=empty}
     \begin{center}
        \subfloat[(a)]{%
            \label{fig:frechetx}
            \includegraphics[width=0.8\columnwidth]{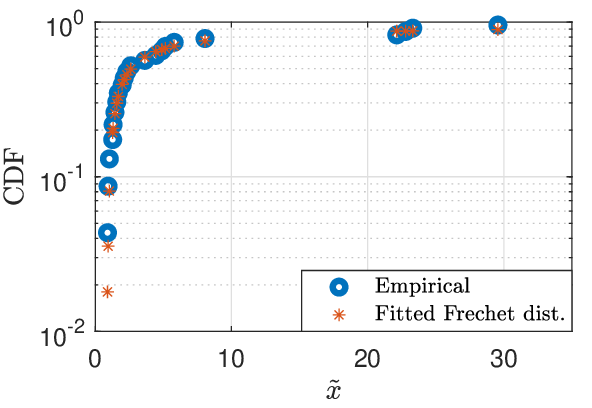}
        }\\%
        \subfloat[(b)]{%
            \label{fig:frechety}
            \includegraphics[width=0.8\columnwidth]{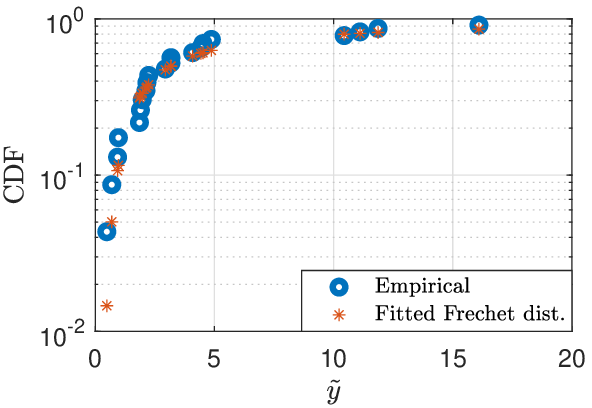}
        }\\
    \end{center}
    \caption{CDF of the received power tail samples transformed based on the Fr\'echet transformation: (a) for the transformed samples of receiver Rx$1$, $x \rightarrow \tilde{x}$, and (b) for the transformed samples of receiver Rx$2$, $y \rightarrow \tilde{y}$.}
   \label{fig:frechet}
\end{figure}

Fig.~\ref{fig:compMGPD} shows the proposed BGPD model determined based on the logistic distribution approach, $G_l(\tilde{x},\tilde{y})$. Empirical joint CDF of the tail samples has also been depicted for better compatibility assessment of the logistic-based BGPD. The BGPD model obtained by the logistic distribution approach can estimate the empirical joint CDF with RMSE $0.8655$. Additionally, the estimation accuracy increases as we approach more extreme values, which confirms the ability of the proposed methodology to model the worse extreme values. However, the exact pattern of the empirical CDF has not been preserved by the logistic-based BGPD model.

\begin{figure}[ht]
    \centering
    \includegraphics[width=0.9\linewidth]{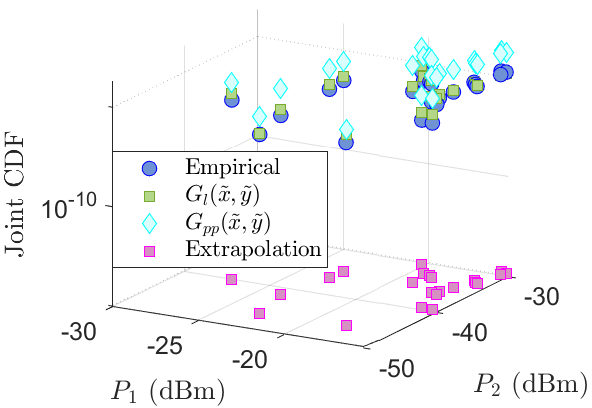}
    \caption{Joint bi-variate CDF of the normalized tail power of empirical, extrapolation-based, logistic-based BGPD, and Poisson point process-based BGPD.}
    \label{fig:compMGPD}
    \end{figure}
\subsection{BGPD Based on the Poisson Point Process Approach}
\label{sec:performance_BGPD_pointprocess}
Fig.~\ref{fig:compMGPD} also depicts the proposed BGPD model determined based on the Poisson point process approach.
It can be seen that the estimated joint CDF based on the Poisson point process approach is in good agreement with the corresponding empirical CDF with an RMSE of $0.8686$. Although the Poisson point process-based approach for modeling the bi-variate tail distribution seems to have the same performance as the logistic distribution-based approach according to the RMSE results, it preserves the shape of the empirical joint CDF very well. This is mainly due to the higher complexity of the additional transformation step as well as integral computation for determining the density of the Poisson point process $\Lambda(A)$.

\subsection{Model Validity Assessment}
\label{sec:performanceValidation}
Figs.~\ref{fig:c0} and \ref{fig:pickandsuniform} are utilized to determine the optimum cut-off point $r_0$ for which the Pickands constraints are satisfied and $H_l(\omega)$ is defined on $r>r_0$. Fig.~\ref{fig:c0} illustrates the $r$-$\omega$ plot of the Poisson point process approach for determining the optimum cut-off point $r_0$. In this figure, the radial component of the Pickands coordinate is plotted versus the corresponding angular component, where $r$ and $\omega$ are the pseudo-polar Pickands coordinates of $\tilde{x}$ and $\tilde{y}$. Please note that $\{(r,\omega)|r<r_0\}$ correspond to independent pairs and $\{(r,\omega)|r>r_0\}$ are matched to dependent pairs of $r$ and $\omega$. According to this plot, $r_0 \approx -0.47$ is the minimum radial component above which the correlation coefficient of Pickands radial and angular components is less than the critical value $0.05$, resulting in independent $r$ and $\omega$ pairs. Therefore, $r_0 \approx -0.47$ is the cut-off radial component that satisfies the first Pickands constraint.

Fig.~\ref{fig:pickandsuniform} shows the distribution of the radial component of the Pickands coordinate for $r>r_0$ where $r$ is the pseudo-polar Pickands radial coordinate of $\tilde{x}$ and $\tilde{y}$, and $r_0 \approx -0.47$ is the optimum radial cut-off point captured from Fig.~\ref{fig:c0}. 
It is observed that the distribution of the radial components $r>r_0$, i.e., those radial components independent from their angular component pairs, fits the uniform distribution, which means that $P(r>R) = \frac{R}{r_0}$, and results in satisfying the second constraint on Pickands coordinates. 

\begin{figure}[ht]
    \centering
    \includegraphics[width=0.8\linewidth]{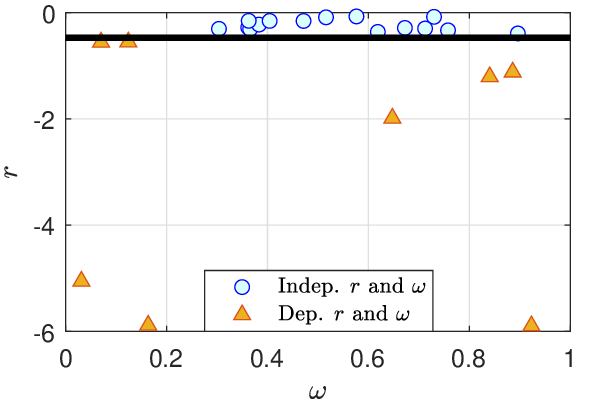}
    \caption{$r$-$\omega$ plot. The Horizontal black line corresponds to $r_0 \approx -0.47$.}
    \label{fig:c0}
\end{figure}

\begin{figure}[ht]
    \centering
    \includegraphics[width=0.8\linewidth]{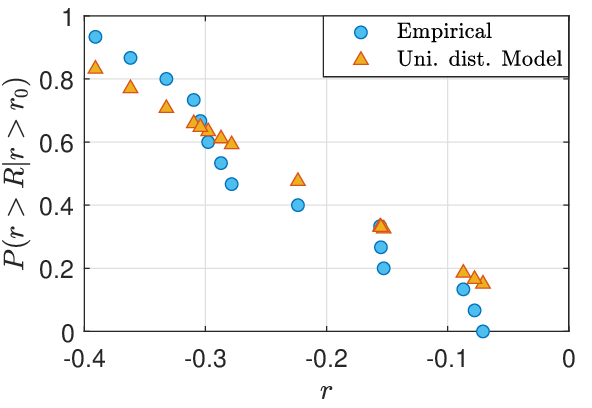}
    \caption{The distribution of the radial component of the Pickands coordinate for $r>r_0$.}
    \label{fig:pickandsuniform}
\end{figure}

We assess the validity of the proposed BGPD models shown in Fig.~\ref{fig:compMGPD} based on the CDF results of Pickands probability density functions $h_{l}(\omega)$ and $h_{pp}(\omega)$ obtained from the logistic-based for $r>r_0$, and Poisson point process-based BGPD approaches, respectively. Accordingly, it is observed that the CDF of the estimated density function of the angular Pickands coordinate $\omega$, for $h_{l}(\omega)$, the $0.5$ mean constraint in (\ref{eqn:theorem_H}) is satisfied as $\int_{\omega} \omega H_l(d\omega) = 0.5001$, indicating the validity of the BGPD model based on the logistic distribution approach. On the other hand, referring to the Pickands density function $h_{pp}(\omega)$ obtained from the Poisson point process-based BGPD approach, the $0.5$ mean constraint is satisfied as $\int_{\omega} \omega H_{pp}(d\omega) = 0.5064$ for the derived $H_{pp}(\omega)$.

\subsection{Comparison with Conventional Tail Estimation Method}
\label{sec:performanceComparison}
Fig.~\ref{fig:compMGPD} additionally illustrates the CDF of the normalized power values for the traditional extrapolation-based approach compared to the empirical and the proposed BGPD models based on the logistic distribution and Poisson point process approaches. 
To derive the traditional extrapolation-based model, we first extract the observations from Rx$1$ and Rx$2$ corresponding to the reliability order of $10^{-3}$-$10^{0}$ PER. Then, upon fitting a variety of distributions to the samples in each observation sequence, the best-fitting distribution is determined as the Gaussian distribution according to the Akaike information criterion/Bayesian information criterion (AIC/BIC) metric. It is worth noting that Gaussian and Rician distribution had almost the same AIC/BIC for both sequences. However, for the sake of simplicity and without loss of generality, we opt to continue with the Gaussian distribution in our analysis as its joint PDF calculation is more straightforward. Then, we determine the joint bi-variate CDF of the Gaussian distribution for variables $X$ and $Y$, corresponding to the received power values obtained from Rx$1$ and Rx$2$, respectively.
Upon determining the joint CDF of $X$ and $Y$, we extrapolate it towards the ultra-reliable region based on the tail approximation approach in \cite{urllc_02}.
According to the results illustrated in Fig.~\ref{fig:compMGPD}, the proposed BGPD models outperform the extrapolated Gaussian model remarkably for a URC system, where the reliability order is in the range of $10^{-9}$-$10^{-5}$. The RMSE of the extrapolated bi-variate Gaussian model is $7.8177 \times 10^{09}$, indicating the superiority of the proposed model with a significant improvement in the estimation of the tail distribution for the ultra-reliability region, where the reliability orders are in the range of $10^{-9}$-$10^{-5}$.
\section{Conclusions}
\label{sec:conclusions}
In this paper, we introduce a novel framework based on the extreme value theory with the goal of estimating the multivariate channel tail statistics for a MIMO-URC. 
The proposed methodology utilizes MEVT to model the tail of the joint probability distribution by using two methods for modeling bi-variate extreme values: the logistic distribution approach, which uses the logistic distribution to determine dependency factors and obtain a model, and the Poisson point process approach, which estimates the probability measure function of the Pickands angular component to model the bi-variate extreme values.
The proposed multi-dimensional channel modeling methodology achieves a significantly better fit to the empirical data in the lower tail than the conventional extrapolation-based approach.
In the future, we are planning to extend this work by determining the estimated BGPD parameters and optimum transmission rate for a delay-constrained real-time communication system leading to the design of a MIMO-URLLC system.
Since the derivation of the statistics requires a large amount of data, such real-time system design requires either the inclusion of confidence intervals in the parameter estimation or the adoption of transfer learning techniques, such as knowledge-assisted training, where the digital twin of a real network is built with network topology, channel and queueing models for offline training, which is then fine-tuned in the real environment with a smaller amount of data.
\vspace{-0.5cm}
\balance 

\ifCLASSOPTIONcaptionsoff
  \newpage
\fi
\bibliographystyle{ieeetr}
\bibliography{MultivariateEVT_ChannledModeling.bib}

\begin{thebibliography}{10}

\bibitem{interface_01}
P.~Popovski, J.~J. Nielsen, C.~Stefanovic, E.~De~Carvalho, E.~Strom, K.~F. Trillingsgaard, A.-S. Bana, D.~M. Kim, R.~Kotaba, J.~Park, {\em et~al.}, ``Wireless access for ultra-reliable low-latency communication: Principles and building blocks,'' {\em IEEE Network}, vol.~32, no.~2, pp.~16--23, 2018.

\bibitem{5G_01}
C.~{Gustafson}, K.~{Haneda}, S.~{Wyne}, and F.~{Tufvesson}, ``On mm-wave multipath clustering and channel modeling,'' {\em IEEE Transactions on Antennas and Propagation}, vol.~62, no.~3, pp.~1445--1455, 2014.

\bibitem{interface_02}
P.~Popovski, {\v{C}}.~Stefanovi{\'c}, J.~J. Nielsen, E.~De~Carvalho, M.~Angjelichinoski, K.~F. Trillingsgaard, and A.-S. Bana, ``Wireless access in ultra-reliable low-latency communication ({URLLC}),'' {\em IEEE Transactions on Communications}, vol.~67, no.~8, pp.~5783--5801, 2019.

\bibitem{urllc_02}
P.~C.~F. {Eggers}, M.~{Angjelichinoski}, and P.~{Popovski}, ``Wireless channel modeling perspectives for ultra-reliable communications,'' {\em IEEE Transactions on Wireless Communications}, vol.~18, no.~4, pp.~2229--2243, 2019.

\bibitem{urllc_05}
M.~Angjelichinoski, K.~F. Trillingsgaard, and P.~Popovski, ``A statistical learning approach to ultra-reliable low latency communication,'' {\em IEEE Transactions on Communications}, vol.~67, no.~7, pp.~5153--5166, 2019.

\bibitem{reliability_01}
S.~Samarakoon, M.~Bennis, W.~Saad, and M.~Debbah, ``Predictive ultra-reliable communication: A survival analysis perspective,'' {\em IEEE Communications Letters}, vol.~25, no.~4, pp.~1221--1225, 2020.

\bibitem{urc_03}
D.~Feng, L.~Lai, J.~Luo, Y.~Zhong, C.~Zheng, and K.~Ying, ``Ultra-reliable and low-latency communications: applications, opportunities and challenges,'' {\em Science China Information Sciences}, vol.~64, pp.~1--12, 2021.

\bibitem{urc_04}
M.~A. Siddiqi, H.~Yu, and J.~Joung, ``5g ultra-reliable low-latency communication implementation challenges and operational issues with iot devices,'' {\em Electronics}, vol.~8, no.~9, p.~981, 2019.

\bibitem{urc_05}
A.~E. Kalor and P.~Popovski, ``Ultra-reliable communication for services with heterogeneous latency requirements,'' in {\em 2019 IEEE Globecom Workshops (GC Wkshps)}, pp.~1--6, 2019.

\bibitem{urc_06}
B.~Kharel, O.~L.~A. L{\'o}pez, H.~Alves, and M.~Latva-Aho, ``Ultra-reliable communication for critical machine type communication via cran-enabled multi-connectivity diversity schemes,'' {\em Sensors}, vol.~21, no.~23, p.~8064, 2021.

\bibitem{MehrniaTWC}
N.~Mehrnia and S.~Coleri, ``Wireless channel modeling based on extreme value theory for ultra-reliable communications,'' {\em IEEE Transactions on Wireless Communications}, pp.~1--1, 2021.

\bibitem{MehrniaTVT}
N.~Mehrnia and S.~Coleri, ``Non-stationary wireless channel modeling approach based on extreme value theory for ultra-reliable communications,'' {\em IEEE Transactions on Vehicular Technology}, vol.~70, no.~8, pp.~8264--8268, 2021.

\bibitem{urllc_diversity_01}
S.~R. Khosravirad, H.~Viswanathan, and W.~Yu, ``Exploiting diversity for ultra-reliable and low-latency wireless control,'' {\em IEEE Transactions on Wireless Communications}, vol.~20, no.~1, pp.~316--331, 2020.

\bibitem{evt_04}
S.~Coles, J.~Bawa, L.~Trenner, and P.~Dorazio, {\em An introduction to statistical modeling of extreme values}, vol.~208.
\newblock Springer, 2001.

\bibitem{timediv_01}
M.~Serror, C.~Dombrowski, K.~Wehrle, and J.~Gross, ``Channel coding versus cooperative arq: Reducing outage probability in ultra-low latency wireless communications,'' in {\em 2015 IEEE Globecom Workshops (GC Wkshps)}, pp.~1--6, IEEE, 2015.

\bibitem{urllc_diversity_02}
C.~Boyd, R.~Kotaba, O.~Tirkkonen, and P.~Popovski, ``Non-orthogonal contention-based access for urllc devices with frequency diversity,'' in {\em 2019 IEEE 20th International Workshop on Signal Processing Advances in Wireless Communications (SPAWC)}, pp.~1--5, IEEE, 2019.

\bibitem{interferencediv}
J.~J. Nielsen, R.~Liu, and P.~Popovski, ``Ultra-reliable low latency communication using interface diversity,'' {\em IEEE Transactions on Communications}, vol.~66, pp.~1322--1334, March 2018.

\bibitem{urllc_08}
V.~N. Swamy, P.~Rigge, G.~Ranade, B.~Nikoli{\'c}, and A.~Sahai, ``Wireless channel dynamics and robustness for ultra-reliable low-latency communications,'' {\em IEEE Journal on Selected Areas in Communications}, vol.~37, no.~4, pp.~705--720, 2019.

\bibitem{urllc_diversity_03}
J.~W. Wallace and M.~A. Jensen, ``Modeling the indoor mimo wireless channel,'' {\em IEEE Transactions on Antennas and Propagation}, vol.~50, no.~5, pp.~591--599, 2002.

\bibitem{urllc_diversity_05}
K.~Zheng, S.~Ou, and X.~Yin, ``Massive mimo channel models: A survey,'' {\em International Journal of Antennas and Propagation}, vol.~2014, 2014.

\bibitem{urllc_diversity_04}
P.~Almers, E.~Bonek, A.~Burr, N.~Czink, M.~Debbah, V.~Degli-Esposti, H.~Hofstetter, P.~Ky{\"o}sti, D.~Laurenson, G.~Matz, {\em et~al.}, ``Survey of channel and radio propagation models for wireless mimo systems,'' {\em EURASIP Journal on Wireless Communications and Networking}, vol.~2007, pp.~1--19, 2007.

\bibitem{urllc_07}
V.~N. Swamy, P.~Rigge, G.~Ranade, B.~Nikolic, and A.~Sahai, ``Predicting wireless channels for ultra-reliable low-latency communications,'' in {\em 2018 IEEE International Symposium on Information Theory (ISIT)}, pp.~2609--2613, IEEE, 2018.

\bibitem{confidenceinterval_03}
W.~Zhang, M.~Derakhshani, and S.~Lambotharan, ``Non-parametric statistical learning for urllc transmission rate control,'' in {\em ICC 2021-IEEE International Conference on Communications}, pp.~1--6, IEEE, 2021.

\bibitem{evt_03}
S.~Kotz and S.~Nadarajah, {\em Extreme value distributions: theory and applications}.
\newblock World Scientific, 2000.

\bibitem{upperlayer_03}
S.~Samarakoon, M.~Bennis, W.~Saad, and M.~Debbah, ``Federated learning for ultra-reliable low-latency {V2V} communications,'' in {\em 2018 IEEE Global Communications Conference (GLOBECOM)}, pp.~1--7, IEEE, 2018.

\bibitem{upperlayer_05}
C.-F. Liu and M.~Bennis, ``Ultra-reliable and low-latency vehicular transmission: An extreme value theory approach,'' {\em IEEE Communications Letters}, vol.~22, no.~6, pp.~1292--1295, 2018.

\bibitem{upperlayer_11}
M.~K. Abdel-Aziz, S.~Samarakoon, C.-F. Liu, M.~Bennis, and W.~Saad, ``Optimized age of information tail for ultra-reliable low-latency communications in vehicular networks,'' {\em IEEE Transactions on Communications}, vol.~68, no.~3, pp.~1911--1924, 2019.

\bibitem{upperlayer_04}
A.~Mahmood and R.~J{\"a}ntti, ``Packet error rate analysis of uncoded schemes in block-fading channels using extreme value theory,'' {\em IEEE Communications Letters}, vol.~21, no.~1, pp.~208--211, 2016.

\bibitem{upperlayer_02}
G.~Song and Y.~Li, ``Asymptotic throughput analysis for channel-aware scheduling,'' {\em IEEE Transactions on Communications}, vol.~54, no.~10, pp.~1827--1834, 2006.

\bibitem{upperlayer_01}
Y.~H. Al-Badarneh, C.~N. Georghiades, and M.-S. Alouini, ``Asymptotic performance analysis of the $ k $ th best link selection over wireless fading channels: An extreme value theory approach,'' {\em IEEE Transactions on Vehicular Technology}, vol.~67, no.~7, pp.~6652--6657, 2018.

\bibitem{upperlayer_12}
Y.~Zhu, Y.~Hu, T.~Yang, T.~Yang, J.~Vogt, and A.~Schmeink, ``Reliability-optimal offloading in low-latency edge computing networks: Analytical and reinforcement learning based designs,'' {\em IEEE Transactions on Vehicular Technology}, vol.~70, no.~6, pp.~6058--6072, 2021.

\bibitem{urllc_10}
C.~Chaccour, M.~N. Soorki, W.~Saad, M.~Bennis, and P.~Popovski, ``Can terahertz provide high-rate reliable low-latency communications for wireless {VR}?,'' {\em IEEE Internet of Things Journal}, vol.~9, no.~12, pp.~9712--9729, 2022.

\bibitem{upperlayer_10}
A.~Subhash and S.~Kalyani, ``Cooperative relaying in a swipt network: Asymptotic analysis using extreme value theory for non-identically distributed rvs,'' {\em IEEE Transactions on Communications}, vol.~69, no.~7, pp.~4360--4372, 2021.

\bibitem{upperlayer_13}
S.~Samarakoon, M.~Bennis, W.~Saad, and M.~Debbah, ``Distributed federated learning for ultra-reliable low-latency vehicular communications,'' {\em IEEE Transactions on Communications}, vol.~68, no.~2, pp.~1146--1159, 2019.

\bibitem{MehrniaTVTRate}
N.~Mehrnia and S.~Coleri, ``Extreme value theory based rate selection for ultra-reliable communications,'' {\em IEEE Transactions on Vehicular Technology}, pp.~1--1, 2022.

\bibitem{MehrniaEucnc}
N.~Mehrnia and S.~Coleri, ``Incorporation of confidence interval into rate selection based on the extreme value theory for ultra-reliable communications,'' in {\em European Conference on Networks and Communications (EuCNC) and the 6G Summit}, pp.~1--6, IEEE, 2022.

\bibitem{evt_01}
B.~Finkenstadt and H.~Rootz{\'e}n, {\em Extreme values in finance, telecommunications, and the environment}.
\newblock CRC Press, 2003.

\bibitem{evt_02}
J.~Galambos, ``Extreme value theory for applications,'' in {\em Proceedings of the Conference on Extreme Value Theory and Applications}, pp.~1--14, Springer, 1994.

\bibitem{evt_05}
Y.~Bensalah, {\em Steps in applying extreme value theory to finance: a review}.
\newblock Bank of Canada, 2000.

\bibitem{evt_06}
M.~Makarov, ``Applications of exact extreme value theorem,'' {\em Journal of Operational Risk}, vol.~2, no.~1, pp.~115--120, 2007.

\bibitem{evt_07}
S.~Caires, ``Extreme value analysis: wave data,'' {\em Joint Tech Comm for Oceanography $\&$ Marine Meteorology (JCOMM) Technical Report}, no.~57, 2011.

\bibitem{evt_08}
A.~Mata, ``Parameter uncertainty for extreme value distributions,'' in {\em GIRO Convention Papers}, pp.~151--173, 2000.

\bibitem{urllc_11}
G.~Durisi, T.~Koch, and P.~Popovski, ``Toward massive, ultrareliable, and low-latency wireless communication with short packets,'' {\em Proceedings of the IEEE}, vol.~104, no.~9, pp.~1711--1726, 2016.

\bibitem{urllc_12}
Y.~Polyanskiy, H.~V. Poor, and S.~Verd{\'u}, ``Channel coding rate in the finite blocklength regime,'' {\em IEEE Transactions on Information Theory}, vol.~56, no.~5, pp.~2307--2359, 2010.

\bibitem{urllc_01}
M.~Bennis, M.~Debbah, and H.~V. Poor, ``Ultra reliable and low-latency wireless communication: Tail, risk, and scale,'' {\em Proceedings of the IEEE}, vol.~106, no.~10, pp.~1834--1853, 2018.

\bibitem{evt_11}
P.~de~Zea~Bermudez and S.~Kotz, ``Parameter estimation of the generalized pareto distribution—part {II},'' {\em Journal of Statistical Planning and Inference}, vol.~140, no.~6, pp.~1374--1388, 2010.

\bibitem{mevt_11}
H.~Joe, R.~L. Smith, and I.~Weissman, ``Bivariate threshold methods for extremes,'' {\em Journal of the royal statistical society: series B (methodological)}, vol.~54, no.~1, pp.~171--183, 1992.

\bibitem{bgpd_01}
L.~Zheng, K.~Ismail, T.~Sayed, and T.~Fatema, ``Bivariate extreme value modeling for road safety estimation,'' {\em Accident Analysis \& Prevention}, vol.~120, pp.~83--91, 2018.

\bibitem{mevt_10}
R.~Michel, {\em Simulation and estimation in multivariate generalized Pareto models}.
\newblock PhD thesis, Universit{\"a}t W{\"u}rzburg, 2006.

\bibitem{mevt_12}
M.~Falk, J.~H{\"u}sler, and R.-D. Reiss, {\em Laws of small numbers: extremes and rare events}.
\newblock Springer Science \& Business Media, 2010.

\bibitem{pickands_01}
P.~Cap{\'e}ra{\`a} and A.-L. Foug{\`e}res, ``Estimation of a bivariate extreme value distribution,'' {\em Extremes}, vol.~3, no.~4, pp.~311--329, 2000.

\bibitem{pickands_02}
D.~Cooley, R.~A. Davis, and P.~Naveau, ``The pairwise beta distribution: A flexible parametric multivariate model for extremes,'' {\em Journal of Multivariate Analysis}, vol.~101, no.~9, pp.~2103--2117, 2010.

\bibitem{pickands_03}
S.~Nadarajah, ``Approximations for bivariate extreme values,'' {\em Extremes}, vol.~3, no.~1, pp.~87--98, 2000.

\bibitem{corcoef_01}
J.~Shi, D.~Anzai, and J.~Wang, ``Channel modeling and performance analysis of diversity reception for implant uwb wireless link,'' {\em IEICE transactions on communications}, vol.~95, no.~10, pp.~3197--3205, 2012.

\bibitem{vehicular_01}
C.~U. Bas and S.~Coleri~Ergen, ``Ultra-wideband channel model for intra-vehicular wireless sensor networks beneath the chassis: From statistical model to simulations,'' {\em IEEE Transactions on Vehicular Technology}, vol.~62, no.~1, pp.~14--25, 2012.

\bibitem{vehicular_02}
U.~Demir, C.~U. Bas, and S.~Coleri~Ergen, ``Engine compartment {UWB} channel model for intra-vehicular wireless sensor networks,'' {\em IEEE Transactions on Vehicular Technology}, vol.~63, no.~6, pp.~2497--2505, 2013.

\end{thebibliography}

\begin{IEEEbiography}[{\includegraphics[width=1in,height=1.25in,clip,keepaspectratio]{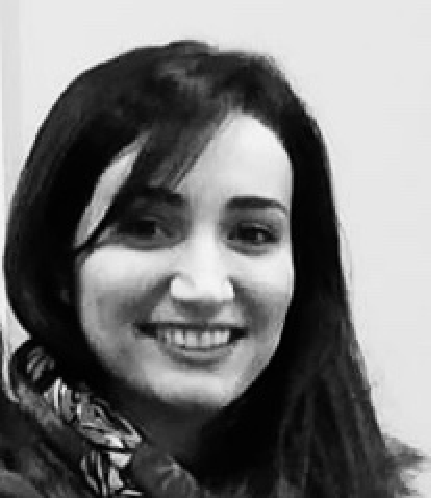}}]%
{Niloofar Mehrnia}
received a B.S. degree in biomedical engineering and M.S. and Ph.D. degree in electrical and electronics engineering from Isfahan University, Istanbul Sehir University, and Koc University in 2014, 2017, and 2022 respectively. From 2019, she also worked as a researcher at Koc University Ford Otosan Automotive Technologies Laboratory (KUFOTAL). Her research interests include wireless channel modeling, ultra-reliable and low-latency communications (URLLC), extreme value theory (EVT), and vehicular communications.
\end{IEEEbiography}

\begin{IEEEbiography}[{\includegraphics[width=1in,height=1.25in,clip,keepaspectratio]{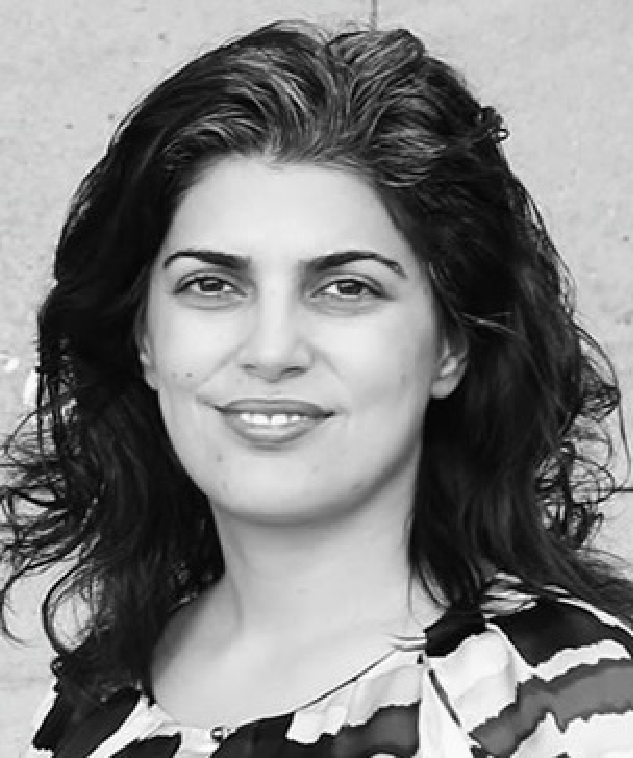}}]{Sinem Coleri}
is a Professor and the Chair of the Department of Electrical and Electronics Engineering at Koc University. She is also the founding director of Wireless Networks Laboratory (WNL) and director of Ford Otosan Automotive Technologies Laboratory. Sinem Coleri received the BS degree in electrical and electronics engineering from Bilkent University in 2000, the M.S. and Ph.D. degrees in electrical engineering and computer sciences from University of California Berkeley in 2002 and 2005. She worked as a research scientist in Wireless Sensor Networks Berkeley Lab under sponsorship of Pirelli and Telecom Italia from 2006 to 2009. Since September 2009, she has been a faculty member in the department of Electrical and Electronics Engineering at Koc University. Her research interests are in 6G wireless communications and networking, machine learning for wireless networks, machine-to-machine communications, wireless networked control systems and vehicular networks. She has received numerous awards and recognitions, including N2Women: Stars in Computer Networking and Communications, TUBITAK (The Scientific and Technological Research Council of Turkey) Incentive Award and IEEE Vehicular Technology Society Neal Shepherd Memorial Best Propagation Paper Award. Dr. Coleri has been Interim Editor-in-Chief of IEEE Open Journal of the Communications Society since 2023, Executive Editor of IEEE Communications Letters since 2023, Editor-at-Large of IEEE Transactions on Communications since 2023, Senior Editor of IEEE Access since 2022 and Editor of IEEE Transactions on Machine Learning in Communications and Networking since 2022. Dr. Coleri is an IEEE Fellow and IEEE ComSoc Distinguished Lecturer.
\end{IEEEbiography}

\end{document}